\documentclass[12pt]{article}
\usepackage[dvips]{graphicx}
\usepackage{amsmath,amsfonts,amssymb,amsthm,color}
\usepackage{natbib}

\newtheorem{theorem}{Theorem}

\newtheorem{example}{Example}
\newtheorem{algorithm}{Algorithm}

\usepackage{bm}
\bmdefine{\Bt}{t}
\bmdefine{\BX}{X}
\bmdefine{\BY}{Y}
\bmdefine{\BZ}{Z}
\bmdefine{\BB}{B}
\bmdefine{\BM}{M}
\bmdefine{\BD}{D}
\bmdefine{\Bi}{i}
\bmdefine{\Bj}{j}
\bmdefine{\Bx}{x}
\bmdefine{\By}{y}
\bmdefine{\Bz}{z}
\bmdefine{\Bv}{v}
\bmdefine{\Bw}{w}
\bmdefine{\Bn}{n}
\bmdefine{\Ba}{a}
\bmdefine{\Bb}{b}
\bmdefine{\Bc}{c}
\bmdefine{\Be}{e}
\bmdefine{\Bu}{u}
\bmdefine{\Bp}{p}
\bmdefine{\Bzero}{0}
\bmdefine{\Bone}{1}

\title{A Localization Approach to Improve Iterative Proportional Scaling in 
Gaussian Graphical Models}
\author{
Hisayuki Hara\\
Department of 
Technology Management for Innovation\\
University of Tokyo \smallskip\\
Akimichi Takemura\\
Graduate School of Information Science and Technology\\
University of Tokyo}
\date{May 2008}

\begin{document}
\maketitle

\begin{abstract}
 We discuss an efficient implementation of the iterative proportional
 scaling procedure in the multivariate Gaussian graphical models. 
 We show that the computational cost can be reduced by localization of
 the update procedure in each iterative step by using the structure of
 a decomposable model obtained by triangulation of the graph associated
 with the model.   
 Some numerical experiments demonstrate the competitive performance of
 the proposed algorithm. 
\end{abstract}

\section{Introduction}
Since Dempster \cite{Dempster} introduced 
a multivariate Gaussian graphical model, also called a
covariance selection model, 
it has been investigated by many authors 
from both theoretical and practical viewpoints. 
On the theory of a Gaussian graphical model, 
see e.g.\ Whittaker \cite{Whittaker}, Lauritzen \cite{lauritzen1996}, 
Cox and Wermuth \cite{Cox-Wermuth} and Edwards \cite{Edwards}.
In recent years  much effort has been devoted to 
application of the Gaussian  graphical model to identify sparse large
network systems, especially genetic networks (e.g.\ \cite{Dobra},
\cite{Li-Gui}, \cite{Drton-2}),  
and the efficient implementation of the inference in the model 
has been extensively studied. 
In this article we discuss an efficient algorithm to compute the
maximum likelihood estimator (MLE) of the covariance matrix in the
Gaussian graphical models. 

When the graph associated with the model is a chordal graph, 
the model is called a decomposable model.
For a decomposable model, the MLE of the covariance matrix is
explicitly obtained.   
For general graphical models other than decomposable models, however, 
we need some iterative procedure to obtain the MLE. The iterative
proportional scaling (IPS) procedure is one of popular 
algorithms to compute the MLE.

The IPS was first introduced by Deming and Stephan \cite{Deming-Stephan} to
estimate cell probabilities in contingency tables subject to certain
fixed marginals. Its convergence and statistical properties have
been well studied by many authors 
(e.g. \cite{Ireland-Kullback}, \cite{Fienberg}) 
and the IPS have been justified in a more general framework
(\cite{Csiszar}). 
Speed and Kiiveri \cite{Speed-Kiiveri} first formulated the IPS 
in a Gaussian graphical model and gave a proof of its convergence. 

However, from a practically point of view, a straightforward application
of the IPS is often computationally too expensive for larger models. 
In the contingency tables several techniques have been developed to
reduce both storage and computational time of the IPS (e.g. \cite{Jirousek},
\cite{Jirousek-Preucil}). 
Badsberg and Malvestuto \cite{Badsberg-Malvestuto} proposed 
a localized implementation of the IPS by using the structure of
decomposable models containing 
the graphical model.
Such a technique is called the chordal extension.
The local computation based on the chordal extension has been a popular 
technique in many fields for numerical computation of a sparse linear
system(e.g. \cite{Rose}, \cite{FKMN}). 

In the present paper 
we describe a localized algorithm based on the chordal extension for
improving the computational efficiency of the IPS in the Gaussian
graphical models.   
Let $\Delta$ be the set of variables which corresponds to the set of
vertices of the graph associated with the model.
The straightforward implementation of the IPS requires approximately 
$O(\vert \Delta \vert^3)$ time in each iterative step for large models. 
In the similar way to the technique in Badsberg and Malvestuto
\cite{Badsberg-Malvestuto},   
we localize the update procedure in each step by using the structure of
a decomposable model containing the model.
The proposed algorithm is shown to require $O(\vert \Delta \vert)$
time for some models. 

The problem of computing the MLE is equivalent to the positive definite
matrix completion problem. 
The proposed algorithm based on the chordal extension is closely related
to the technique discussed by Fukuda et al. \cite{FKMN} in the framework
of the positive definite matrix completion problem but not the same. 

As pointed out in Dahl et al.\ \cite{Dahl}, 
the implementation of the IPS requires enumeration of all maximal
cliques of the graph and this enumeration has an exponential complexity.  
Hence the application of the IPS to large models may be limited.     
However in the case where the model is relatively small or 
the structure of the model is simple, it may be feasible to enumerate
maximal cliques.  
In this article we consider such situations. 

The organization of this paper is as follows. 
In Section 2 we summarize notations and basic facts on graphs and give a
brief review of Gaussian graphical models and the IPS algorithm for
covariance matrices. 
In Section 3 we propose an efficient implementation of the update
procedure of the IPS. 
In Section 4 we perform some numerical experiments to illustrate the
effectiveness of the proposed procedure. 
We end this paper with some concluding remarks in Section 5.

\section{Background and preliminaries}
\subsection{Preliminaries on decompositions of graphs}
In this section we summarize some preliminary facts on decompositions
of graphs needed in the argument of the following sections according
to Leimer \cite{Leimer}, Lauritzen \cite{lauritzen1996} 
and Malvestuto and Moscarini \cite{Malvestuto-Moscarini}. 

Let ${\cal G} = (\Delta, E)$ be an undirected graph, 
where $\Delta$ denotes the set of vertices and $E$ denotes the set of
edges. 
A subset of $\Delta$ which induces a complete subgraph is called a clique
of ${\cal G}$. 
Define the set of maximal cliques of ${\cal G}$ by ${\cal C}$. 
For a subset of vertices $V$, let ${\cal G}(V)$ denote the subgraph of 
${\cal G}$ induced by $V$.  When a graph ${\cal G}$ is not connected,
we can consider each connected component of ${\cal G}$ separately.
Therefore we only consider a connected graph from now on.

A subset $S \subset \Delta$ is said to be a separator of ${\cal G}$ if 
${\cal G}(\Delta \setminus S)$ is disconnected. 
For a separator $S$, a triple $(A,B,S)$ of disjoint
subsets of $\Delta$ such that 
$A \cup B \cup S = \Delta$ is said to form a decomposition of 
${\cal G}$. 
A separator $S$ is called a clique separator if $S$ is a clique of ${\cal G}$. 
For two non-adjacent vertices $\delta$ and $\delta'$, $S\subset \Delta$
is said to be a $(\delta,\delta')$-separator if $\delta \in A$ and 
$\delta' \in B$ for a decomposition $(A, B, S)$. 
A $(\delta,\delta')$-separator which is minimal with respect to 
inclusion relation is called a minimal $(\delta,\delta')$-separator
or a minimal vertex separator(Lauritzen\cite{lauritzen1996}). 
Denote by ${\cal S}$ the set of minimal vertex separators for all 
non-adjacent pairs of vertices in ${\cal G}$. 

A graph ${\cal G}$ is called reducible if $\Delta$ contains a clique
separator and otherwise ${\cal G}$ is said to be prime.
If ${\cal G}(V)$ is prime and ${\cal G}(V')$ is reducible for all 
$V'$ with 
$V \subsetneq V' \subset \Delta$, 
${\cal G}(V)$ is called a maximal prime subgraph (mp-subgraph) of 
${\cal G}$.  
For any reducible graph, its decomposition into mp-subgraphs is
uniquely defined(Leimer\cite{Leimer}, Malvestuto and
Moscarini\cite{Malvestuto-Moscarini}). 
Denote by ${\cal V}$ 
the set of subsets of $\Delta$ which induces mp-subgraphs of 
${\cal G}$ and let $\vert {\cal V} \vert = M$. 
Then there exists a sequence $V_1,\dots,V_M \in {\cal V}$
such that for every $m = 2,\ldots,M$ there exists $m' < m$ with  
$$
V_{m'} \supset V_m \cap (V_1 \cup \cdots \cup V_{m-1}). 
$$
Such a sequence is called a D-ordered sequence. 
Let $S_m := V_m \cap (V_1 \cup \cdots \cup V_{m-1})$ for $m=2,\ldots,M$.
Define 
$\bar {\cal S} = \{ S_2, \dots, S_m \}$. 
Denote by ${\cal S}_C$ the set of clique separators of ${\cal G}$.
Then $\bar {\cal S}$ satisfy $\bar {\cal S} = {\cal S} \cap {\cal S}_C$.
So we call elements of $\bar {\cal S}$ clique minimal vertex separators.
Leimer\cite{Leimer} showed that reducible graphs always have a D-ordered
sequence with $V_1 = V$ for any $V \in {\cal V}$. 
Hence a D-ordered sequence is not uniquely defined.
However $\bar {\cal S}$ is common for all D-ordered sequences.
\begin{example}[A reducible graph]
 The graph ${\cal G}$ in Figure \ref{fig:reducible} is an example of
 reducible graphs. 
 ${\cal G}$ has two clique minimal vertex separators $S_2 :=\{3,4\}$ and
 $S_3:=\{5,6\}$. 
 Define $V_1$, $V_2$ and $V_3$ by 
 $$
 V_1 := \{1,2,3,4\}, \quad 
 V_2 := \{3,4,5,6\}, \quad 
 V_3 := \{5,6,7,8\}
 $$
 as in Figure \ref{fig:reducible}. 
 Then ${\cal V}=\{V_1,V_2,V_3\}$ and 
 the sequence $V_1$, $V_2$, $V_3$ is a D-ordered sequence.
\begin{figure}[htbp]
 \centering
 \includegraphics{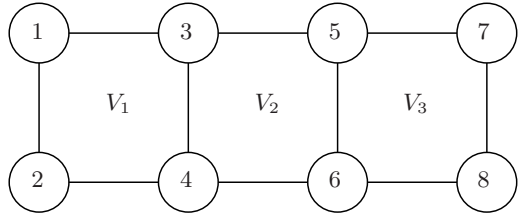} 
 \caption{A reducible graph with eight vertices}
 \label{fig:reducible}
\end{figure}
\end{example}
When ${\cal G}$ is a chordal graph, ${\cal V}$ and 
$\bar {\cal S}$ are equal to 
the set of maximal cliques ${\cal C}$ and 
the set of minimal vertex separators ${\cal S}$ of ${\cal G}$,
respectively. 
Hence $\vert {\cal C} \vert =M$. 
A D-ordered sequence for a chordal graph is called 
a perfect sequence of maximal cliques. 
There exists a perfect sequence of maximal cliques 
$C_1,\ldots,C_M$ such that $C_1=C$ for any 
$C \in {\cal C}$(e.g. Lauritzen\cite{lauritzen1996}).

For a vertex $\delta \in \Delta$, let $\mathrm{adj}(\delta)$ denote the
set of vertices adjacent to $\delta$. 
When $\mathrm{adj}(\delta)$ is a clique, $\delta$ is called a simplicial
vertex. 
A simplicial vertex is contained in only one maximal clique $C$.
Hence if $\delta$ is simplicial and $\delta \in C$, 
then 
$\mathrm{adj}(\delta) = C \setminus \{\delta\}$.
A sequence of vertices 
$\delta_1,\delta_2,\ldots,\delta_{\vert \Delta \vert}$
is called a perfect elimination order of vertices of ${\cal G}$ 
if $\delta_i$ is a simplicial vertex in
${\cal G}(\bigcup_{j=i}^{\vert \Delta \vert}\{\delta_j\})$.  
It is well known that ${\cal G}$ is a chordal graph if and only if 
${\cal G}$ possesses  a perfect elimination order (Dirac\cite{Dirac}). 
Let $C_1,\ldots,C_M$ be a perfect sequence of maximal cliques of 
a chordal graph ${\cal G}$. 
Define $R_1 := C_1 \setminus S_2$, 
$S_m := C_m \cap (C_1 \cup \cdots \cup C_{m-1})$ and 
$R_m := C_m \setminus S_m$ for $m=2,\ldots,M$. 
Let $r_m := \vert R_m \vert$. 
Let $\delta^m_1,\ldots,\delta^m_{r_m}$ be any sequence of vertices in 
$R_m$.
Then the sequence of vertices 
$$
\delta^M_1,\ldots,\delta^M_{r_M},
\delta^{M-1}_1,\ldots,\delta^{M-1}_{r_{M-1}},\ldots,
\delta^1_1,\ldots,\delta^1_{r_1}
$$
is a perfect elimination order of ${\cal G}$.
We call it a perfect elimination order 
induced by the perfect sequence $C_1,\ldots,C_M$. 

We introduce some notations and a basic formula for matrices
needed in the following sections.  
Let $A = \{a_{ij}\}$ be a $\vert \Delta \vert \times \vert \Delta \vert$
matrix. 
For two subsets $\Delta_1$ and $\Delta_2$ of $\Delta$, 
we let 
$$
A_{\Delta_1 \Delta_2} = \{a_{ij}\}_{i \in \Delta_1, j \in \Delta_2}
$$
denote a $\vert \Delta_1 \vert \times \vert \Delta_2 \vert$ submatrix of
$A$. 
Define 
$$
A_{\Delta_1 \Delta_2}^{-1} := (A^{-1})_{\Delta_1 \Delta_2}.
$$
We let $[A_{\Delta_1 \Delta_2}]^{\Delta}$ denote the 
$\vert \Delta \vert \times \vert \Delta \vert$ matrix such that
$$
([A_{\Delta_1 \Delta_2}]^{\Delta})_{ij} = \left\{
\begin{array}{ll}
 a_{ij} & \text{ if } i \in \Delta_1, j \in \Delta_2\\
 0 & \text{ otherwise }.
\end{array}
\right.
$$
Let $\Delta_2=\Delta_1^C$ and decompose a symmetric matrix $A$
into blocks as
$$
A = \left(
\begin{array}{cc}
 A_{\Delta_1 \Delta_1} &  A_{\Delta_1 \Delta_2}\\
 A^{\prime}_{\Delta_1 \Delta_2} &  A_{\Delta_2 \Delta_2}\\
\end{array}
\right).
$$
Here for notational simplicity we displayed $A$ for the case that
the elements of $\Delta_1$ are smaller than those of $\Delta_2$.
Suppose that $A_{\Delta_2 \Delta_2}$ and 
$
A_{\Delta_1 \Delta_1} - 
A_{\Delta_1 \Delta_2}
(A_{\Delta_2 \Delta_2})^{-1}
A^{\prime}_{\Delta_1 \Delta_2}
$
are both positive definite.
Then $A$ is positive definite and 
\begin{equation}
 \label{formula:matrix}
  A^{-1}_{\Delta_1 \Delta_1} = 
  \left(
   A_{\Delta_1 \Delta_1} - 
   A_{\Delta_1 \Delta_2}
   (A_{\Delta_2 \Delta_2})^{-1}
   A^{\prime}_{\Delta_1 \Delta_2}
  \right)^{-1}.
\end{equation}

\subsection{Gaussian graphical models}
Let ${\cal M}^+({\cal G})$ denote the set of 
$\vert \Delta \vert \times \vert \Delta \vert$ positive definite
matrices $K = \{k_{ij}\}$ 
such that $k_{ij}=0$
for all $i$, $j \in \Delta$
with $i \neq j$ and $(i,j) \notin E$.
Then the Gaussian graphical model for 
$\vert \Delta \vert$ dimensional random variable 
$Y=(Y^{(1)},\ldots,Y^{(\vert \Delta \vert)})'$ 
associated with a graph $\cal G$ 
is defined as
$$
Y \sim N_{\vert \Delta \vert}(\mu, \Sigma), \quad 
K := \Sigma^{-1} \in {\cal M}^+({\cal G}).
$$
$k_{ij}=0$ indicates the conditional independence between 
$Y^{(i)}$ and $Y^{(j)}$ given all other variables. 
In what follows, we identify ${\cal M}^+({\cal G})$ with the
corresponding graphical model. 
Let $y_1,\ldots,y_n$ be i.i.d.\ samples from ${\cal M}^+({\cal G})$.  
Define $\bar{y}$ and $W$ by 
$$
\bar{y} := n^{-1}\sum_{i=1}^n y_i, \quad 
W := \sum_{i=1}^n 
(y_i - \bar{y})(y_i - \bar{y})', 
$$
respectively.
The likelihood equation is written as
$$
L(\mu, K) \propto (\mathrm{det}K)^{n/2}
\exp\left\{  
- \frac{1}{2} \mathrm{tr} KW - 
\frac{n}{2} \mathrm{tr} K (\bar{y} - \mu)(\bar{y} - \mu)'
\right\}.
$$
The MLE of $\mu$ is $\bar{y}$.
The likelihood equations involving $K$ are expressed as
\begin{equation}
 \label{eq:likelihood}
n K^{-1}_{CC} = n \Sigma_{CC} = W_{CC}, \quad \forall C \in {\cal C}.
\end{equation}

For a subset of vertices $V \subset \Delta$, 
let $\hat{K}_{V V}$ denote the MLE of $K$ in the marginal
model associated with the graph ${\cal G}(V)$ 
based on the data in $V$-marginal sample only. 
Let $S$ be a clique separator of ${\cal G}$ and 
$(A, B, S)$ be a decomposition of ${\cal G}$. 
Let $V = A \cup S$ and $V' = B \cup S$. 
Then the MLE $\hat{K}$ is known to satisfy 
\begin{equation}
 \label{localize-1}
 \hat{K} = 
  \left[
   \hat{K}_{VV}
  \right]^{\Delta} 
  +
  \left[
   \hat{K}_{V'V'}
  \right]^{\Delta} 
  -
  n
  \left[
   (W_{S S})^{-1}
  \right]^{\Delta}
\end{equation}
(e.g. Lauritzen \cite{lauritzen1996}). 
More generally, for the set of mp-subgraphs ${\cal V}$ and 
the set of clique minimal vertex separators $\bar {\cal S}$, 
\begin{equation}
 \label{localize}
 \hat{K} = \sum_{V \in {\cal V}}
  \left[
   \hat{K}_{V V}
  \right]^{\Delta} 
  -
  n
  \sum_{S \in \bar {\cal S}}
  \left[
   (W_{S S})^{-1}
  \right]^{\Delta}. 
\end{equation}
As mentioned in the previous section, when the model is decomposable,  
${\cal V} = {\cal C}$ and ${\cal S} = \bar {\cal S}$.
Hence from (\ref{eq:likelihood}), 
$\hat{K}$ is explicitly written by 
$$
\hat{K} = n \sum_{C \in {\cal C}}
\left[
(W_{CC})^{-1}
\right]^{\Delta} 
-
n
\sum_{S \in {\cal S}}
\left[
(W_{S S})^{-1}
\right]^{\Delta}.
$$

However for other graphical models, 
we need some iterative procedure 
for computing the first term on the right-hand side
of  (\ref{localize}).  
The following IPS is commonly used for 
this purpose. 
Note that the second term 
on the right-hand side of  (\ref{localize}) needs to be calculated
only once and is not involved in the iterative procedure.
IPS consists of iteratively and successively adjusting 
$\Sigma_{CC}$ for $C \in {\cal C}$ as in (\ref{eq:likelihood}). 
Let $K^{t}$ and $\Sigma^{t} = (K^{t})^{-1}$ denote the estimated $K$ and
$\Sigma$ at the $t$-th step of iteration, respectively. 
Define $D := \Delta \setminus C$ for $C \in {\cal C}$. 
Then the $t$-th iterative step of the IPS is described by 
the update rule of $K^{t}$ as follows.\\ 
{\bf Algorithm 0} (Iterative proportional scaling for $K$)
 \begin{description}
  \item[Step 0] $t \leftarrow 1$ and select an initial estimate
	     $K^0$ such that $K^0 \in {\cal M}^+({\cal G})$.
  \item[Step 1] Select a maximal clique $C \in {\cal C}$ and update $K$
	     as follows,  
	     \begin{align}
	      \label{IPS}
	      (K^{t})_{CC} & \leftarrow (W_{CC})^{-1}
	      +(K^{t-1})_{CD}((K^{t-1})_{DD})^{-1}(K^{t-1})_{DC}\\
	      (K^{t})_{CD} & \leftarrow (K^{t-1})_{CD}\notag\\
	      (K^{t})_{DD} & \leftarrow (K^{t-1})_{DD}.\notag
	     \end{align}
  \item[Step 2] If $K^{t}$ converges, exit. Otherwise 
	     $t \leftarrow t+1$ and go to Step 1.
 \end{description}
From (\ref{formula:matrix}), 
it is easy to see that 
$$
(K^{t})^{-1}_{CC} = (\Sigma^{t})_{CC} = W_{CC}/n.
$$
In Step 1, only the $C$-marginal of $K$ is updated.
Therefore we note that 
if the initial estimate $K^0$ satisfies 
$K^0 \in {\cal M}^+({\cal G})$,  
$K^t$ satisfies $K^{t} \in {\cal M}^+({\cal G})$ for all $t$. 
By using the argument of Csisz\'ar \cite{Csiszar}, 
the convergence of the algorithm to the MLE
$$
\lim_{n \to \infty} K^{t} = \hat{K}, \quad
\lim_{n \to \infty} \Sigma^{t} = \hat{\Sigma}
$$
is guaranteed (Speed and Kiiveri \cite{Speed-Kiiveri} and Lauritzen
\cite{lauritzen1996}).  

The fact (\ref{localize-1}) suggests that 
the decomposition $(A, B, S)$ for a clique separator $S$
can localize the problem, that is, 
in order to obtain the MLE $\hat{K}$, 
it suffices to compute the MLE of submatrix $\hat{K}_{VV}$ 
and $\hat{K}_{V'V'}$, where $V=A \cup S$ and $V'= B \cup S$. 
Especially if the decomposition by mp-subgraphs is obtained, 
we need only to compute $\hat{K}_{VV}$ for each $V \in {\cal V}$.

From a complexity theoretic point of view, the $t$-th iterative step
(\ref{IPS}) requires 
$O(\vert D \vert^3 + \vert D \vert^2 \vert C \vert + 
\vert D \vert \vert C \vert^2)$ time. 
The graphical model with 
$$
{\cal C} = \{\{1,2\},\{2,3\},\ldots,\{\vert \Delta \vert-1, \vert
\Delta \vert\}, \{\vert \Delta \vert, 1\}\}
$$ 
is called the 
$\vert \Delta \vert$-dimensional cycle model or 
$\vert \Delta \vert$ cycle model.
Note that the cycle is prime. 
In the case of $\vert \Delta \vert$ cycle model, 
$\vert C \vert = 2$ and 
$\vert D \vert = \vert \Delta \vert -2$. 
Hence when $\vert \Delta \vert \ge 4$, 
the iterative step (\ref{IPS}) requires 
$O((\vert \Delta \vert -2 )^3)$ time. 
In the next section we propose a more efficient algorithm for
computing (\ref{IPS}) by using the structure of a chordal extension of a 
graph. 

\section{A localized algorithm of IPS}
From (\ref{formula:matrix}), we note that (\ref{IPS}) is rewritten as
\begin{align}
 \label{IPS-2}
  (K^{t})_{CC} & = 
 (W_{CC})^{-1} + (K^{t-1})_{CC} - ((\Sigma^{t-1})_{CC})^{-1}\notag\\
 & = (W_{CC})^{-1} + (K^{t-1})_{CC} - ((K^{t-1})^{-1}_{CC})^{-1}.
\end{align}
In this section we provide an efficient algorithm to compute 
$((K^{t-1})^{-1}_{CC})^{-1}$ by using the structure of ${\cal G}$. 
For a graph ${\cal G}$, 
let ${\cal G}^*$ be a chordal graph obtained by triangulating 
${\cal G}$.  
Such ${\cal G}^*$ is called a chordal extension of ${\cal G}$.
Figure \ref{figure:5cycle} represents an example of the five cycle model
and its chordal extension. 

\begin{figure}[htbp]
 \centering
 \begin{tabular}{cc}
 \includegraphics{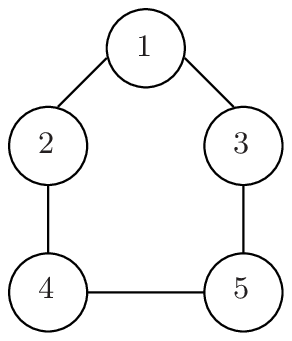} &
 \includegraphics{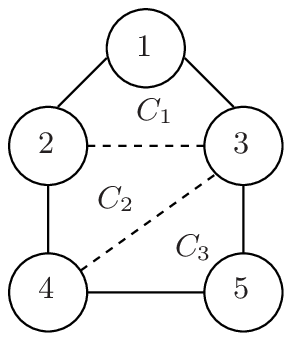}\\
  (i) the five cycle model & (ii) a chordal extension of (i)
 \end{tabular}
 \caption{The five cycle model and its chordal extension} 
 \label{figure:5cycle} 
\end{figure}

Let $C^*_1,\ldots,C^*_M$ be a perfect sequence of the maximal cliques of
${\cal G}^*$ with $C^*_1 \supset C$. 
Let $S^*_m := C^*_m \cap (C^*_1 \cup \cdots \cup C^*_{m-1})$ for 
$m=2,\ldots,M$ be minimal vertex separators of ${\cal G}^*$. 
We propose the following algorithm to 
compute $((K^{t-1})^{-1}_{CC})^{-1}$ for each 
maximal clique $C \in {\cal C}$.

\begin{algorithm}[Computing $((K^{t-1})^{-1}_{CC})^{-1}$]
 \label{alg:1}
 ~\\
 \vspace{-0.5cm}
 {\rm 
 \begin{description}
  \item[Step 0] $m \leftarrow M$ and $K^* \leftarrow K^{t-1}$.
  \item[Step 1] If $m \neq 1$, select a simplicial vertex 
	     $\delta \in C_m^*$ of ${\cal G}^*$.  \\
	     \hspace*{0.35cm}
	     If $m=1$, select a vertex $\delta \notin C$. \\
	     \hspace*{0.35cm}
	     Let $Q = C_m^* \setminus \{\delta\}$.
  \item[Step 2] Update $K^*_{QQ}$ by 
	     \begin{equation}
	      \label{step2}
	       K^*_{QQ} \leftarrow K^*_{QQ} -
	       (k^*_{\delta \delta})^{-1}K^*_{Q \delta}K^*_{\delta Q}.
	     \end{equation}
  \item[Step 3]
	     Update $C^*_m$, ${\cal G}^*$ and $\Delta$ as follows,
	     $$
	     C^*_m \leftarrow Q,\quad 
	      {\cal G}^* \leftarrow {\cal G}^*(\Delta \setminus \{\delta \}), 
	     \quad 
	      \Delta \leftarrow \Delta \setminus \{ \delta \}.
	     $$
	     \hspace*{0.35cm}
	     If $C^*_m = S^*_m$, 
	     $m \leftarrow m-1$.\\
	     \hspace*{0.35cm}
	     If $C^*_m = C$, return $K^*_{CC}$.
	     Otherwise, go to Step 1.
 \end{description}
 }
\end{algorithm}

Now we state the main theorem of this paper.

\begin{theorem}
 The output $K^*_{CC}$ of Algorithm \ref{alg:1} is equal to 
 $((K^{t-1})^{-1}_{CC})^{-1}$. 
\end{theorem}

\begin{proof}
 Let $\delta \in C^*_M$ be a simplicial vertex in ${\cal G}^*$. 
 Define 
 $Q := C^*_M \setminus \{\delta\}$, 
 $Q_1 := \Delta \setminus \{\delta\}$ and 
 $Q_2 := \Delta \setminus C^*_M$.
 Since $\mathrm{adj}(\delta) \subset C^*_M$ and 
 $K^{(t-1)} \in {\cal M}^+({\cal G})$, 
 $(K^{t-1})_{Q_2 \delta} = \bm{0}$. 
 Noting that $Q \cup Q_2 = Q_1$, 
 we have from (\ref{formula:matrix}) 
 \begin{align*}
 ((K^{t-1})_{Q_1 Q_1}^{-1})^{-1} 
  & = (K^{t-1})_{Q_1 Q_1} - 
   (k^{t-1})^{-1}_{\delta \delta} 
   (K^{t-1})_{Q_1 \delta} (K^{t-1})_{\delta Q_1}\\
  & = (K^{t-1})_{Q_1 Q_1} - 
  (k^{t-1})^{-1}_{\delta \delta} 
  \left(
  \begin{array}{c}
   \bm{0} \\
   (K^{t-1})_{Q \delta} 
  \end{array}
  \right)
  \left(
  \begin{array}{cc}
   \bm{0} & (K^{t-1})_{\delta Q}
  \end{array}
 \right)\\
  & = (K^{t-1})_{Q_1 Q_1} - 
 \left(
 \begin{array}{cc}
  \bm{0} & \bm{0}\\
  \bm{0} & 
   (k^{t-1})^{-1}_{\delta \delta} 
   (K^{t-1})_{Q \delta} (K^{t-1})_{\delta Q}
 \end{array}
 \right)\\
 \end{align*}
 and 
 $((K^{t-1})_{Q_1 Q_1}^{-1})^{-1} \in {\cal M}^+({\cal G}(Q_1))$, 
 where 
 $(k^{t-1})_{\delta \delta}$ is the $(\delta,\delta)$-th element of 
 $K^{t-1}$. 
 By iterating the procedure in accordance with the perfect elimination
 order induced by the perfect sequence $C^*_1,\ldots,C^*_M$, we 
 complete the proof. 
\end{proof}

In Algorithm \ref{alg:1}, the triangulation ${\cal G}^*$ is arbitrary.
However for every iterative step of adjusting the $C$-marginal, 
we have to use the perfect sequence with $C_1^* \supset C$. 

\begin{example}[the five cycle model]
 Consider the five cycle model in Figure \ref{figure:5cycle}-(i).
 $K$ is expressed by 
 $$
 K = \left(
 \begin{array}{ccccc}
  k_{11} & k_{12} & k_{13} & 0 & 0\\
  k_{12} & k_{22} & 0 & k_{24} & 0\\
  k_{13} & 0 & k_{33} & 0 & k_{35}\\
  0 & k_{24} & 0 & k_{44} & k_{45}\\
  0 & 0 & k_{35} & k_{45} & k_{55}\\
 \end{array}
 \right).
 $$
 By adding the fill-in edges $\{2,3\}$ and $\{3,4\}$, 
 a triangulated graph ${\cal G}^*$ can be obtained as in Figure
 \ref{figure:5cycle}-(ii). 
 Consider the case where $C = \{1,2\}$. 
 Define $C_1^* = \{1,2,3\}$, $C_2^*=\{2,3,4\}$ and 
 $C_3^* = \{3,4,5\}$.
 Then the sequence $C^*_1,C^*_2,C^*_3$ is perfect and 
 it induces a perfect elimination order $5,4,3,2,1$.
 The update of $K^*$ in step 2 in accordance with the perfect
 elimination order is described as follows,
 \begin{align*}
 K^*_{34,34} 
 & \leftarrow 
 K^*_{34,34} - 
 (k^*_{55})^{-1}
 \left(
 \begin{array}{c}
  k^*_{35}\\
  k^*_{45}\\  
 \end{array}
 \right)
 (k^*_{35} \; k^*_{45}), \\
 K^*_{23,23} 
 & \leftarrow 
 K^*_{23,23} - 
 (k^*_{44})^{-1}
 \left(
 \begin{array}{c}
  k^*_{24}\\
  k^*_{34}\\  
 \end{array}
 \right)
 (k^*_{24} \; k^*_{34}), \\
 K^*_{12,12} 
 & \leftarrow 
 K^*_{12,12} - 
 (k^*_{33})^{-1}
 \left(
 \begin{array}{c}
  k^*_{13}\\
  k^*_{23}\\  
 \end{array}
 \right)
 (k^*_{13} \; k^*_{23}).
 \end{align*}
 Then  
 $K^*_{12,12} = ((K^{t-1})^{-1}_{12,12})^{-1}
 = ((K^{t-1})^{-1}_{CC})^{-1}
 $. 
 \hfill\qed
\end{example}

We now analyze the computational cost of the proposed algorithm. 
In Step 2, the running time of the calculation of (\ref{step2})
is as follows,
\begin{itemize}
 \item $K^*_1 := (k^*_{\delta \delta})^{-1}K^*_{Q \delta}$ requires 
       $\vert Q \vert$ divisions ;
 \item $K^*_2:=K^*_1 K^*_{\delta Q}$
       requires $\vert Q \vert^2$ multiplications ;
 \item $K^*_{QQ} - K^*_2$ requires $\vert Q \vert^2$ subtractions. 
\end{itemize}
Define $R^*_1 := C^*_1 \setminus C$ and 
$R^*_m := C^*_m \setminus S^*_m$ for $m=2,\ldots,M$. 
$\vert Q \vert$ ranges over 
$\{\vert C_m \vert - j \mid  1 \le j \le R^*_m, \ 1 \le m \le M\}$.
Let $\mu$, $\gamma$ and $\sigma$ measure the time units required by a
single multiplication, division and subtraction, respectively.
Then the running time of Algorithm \ref{alg:1} amounts to 
{\allowdisplaybreaks
\begin{align*}
 &(\mu + \sigma) 
 \sum_{m=1}^{M} 
 \sum_{j=1}^{R^*_m} 
 \left(
 \vert C^*_m \vert - j
 \right)^2
 +
 \delta 
 \sum_{m=1}^{M} 
 \sum_{j=1}^{R^*_m} 
 \left(
 \vert C^*_m \vert - j
 \right)\\
 & \qquad =
 (\mu + \sigma) 
 \sum_{m=1}^{M} 
 \Bigr\{
 \vert R^*_m \vert \vert C^*_m \vert^2
 - 
 \vert R^*_m \vert \vert C^*_m \vert
 -
 \vert R^*_m \vert^2 \vert C^*_m \vert \\
 & \qquad \qquad \qquad \qquad \qquad \qquad \qquad
 + 
 \frac{\vert R^*_m \vert (\vert R^*_m \vert +1)
 (2\vert R^*_m \vert + 1)}{6}
 \Bigl\}\\
 & \qquad\qquad + 
 \delta \sum_{m=1}^{M}
 \left\{
 \vert R^*_m \vert \vert C^*_m \vert 
 -
 \frac{(1+ \vert R^*_m \vert) \vert R^*_m \vert}{2}
 \right\}
 + 2 \sigma \vert C \vert^2. 
\end{align*}
}
Since $\vert C^*_m \vert \ge \vert R^*_m \vert$, 
the computational cost of Algorithm \ref{alg:1} is 
$ O\left(\sum_{m=1}^M \vert R^*_m \vert \vert C^*_m \vert^2 \right)$. 
Once $((K^{t-1})^{-1}_{CC})^{-1}$ is obtained, 
$O(\vert C \vert^2)$ additions are required 
to compute (\ref{IPS-2}).
Note that we can compute $(W_{CC})^{-1}$ once 
before the IPS procedure.  Hence 
the computational cost of the $t$-th iterative step amounts to 
$ O\left(\vert C \vert^2 + \sum_{m=1}^M \vert R^*_m \vert \vert
C^*_m \vert^2 \right)$.  
In the case of cycle models, 
$\vert C \vert = 2$, 
$M = \vert \Delta \vert - 2$, 
$\vert C^*_m \vert = 3$ and 
$\vert R^*_m \vert = 1$.
Thus the computational cost is $O(\vert \Delta \vert)$.
As mentioned in the previous section, 
the direct computation of 
$((K^{t-1})^{-1}_{CC})^{-1}$
requires 
$O((\Delta \setminus C)^3)=O(\vert D \vert^3 )$ time and 
in the case of cycle models it requires 
$O((\vert \Delta \vert -2)^3)$ time.
Hence we can see the efficiency of the proposed algorithm.

\section{Numerical experiments for cycle models}
In this section we compare the localized IPS proposed in the previous
section with the direct computation of the IPS by numerical experiments.   
We consider the $\vert \Delta \vert$ cycle models with 
$\vert \Delta \vert = 5,10,50,100,200,300,500,1000$.
We set $K = I_{\vert \Delta \vert}$. 
We generate 100 Wishart matrices 
$W$ with the parameter $I_{\vert \Delta \vert}$ and the degrees of
freedom $\vert \Delta \vert$ and computed the MLE $\hat{K}$ for 
$\vert \Delta \vert$ cycle models by using the proposed algorithm and
the direct computation of the IPS.    
We set the initial estimate $K^0 := I_{\vert \Delta \vert}$.
As a convergence criterion, we used 
$\sum_{i,j} |k^t_{ij}| \le 10^{-6}$. 
The computation was done on a Intel Core 2 Duo 3.0 GHz CPU machine by
using R language. 
Table \ref{table:1} presents the average CPU time per one iterative
step to update $K^{t-1}$ in (\ref{IPS-2}) for both algorithms.

We can see the competitive performance of the proposed algorithm when 
$\vert \Delta \vert =5$ and $\vert \Delta \vert \ge 200$.
However the direct computation is faster than the proposed one for
$\vert \Delta \vert =10$, $50$, $100$.
In the update procedure of direct computation (\ref{IPS-2}), 
the computation of $((K^{t-1})_{DD})^{-1}$ is the most computationally
expensive and in theory it requires $O(\vert D \vert^3)$ time. 
Table \ref{table:2} shows the average CPU time for computing a 
$\vert \Delta \vert \times \vert \Delta \vert$  
inverse matrix by using R language on the same machine. 
As seen from the table, while the CPU time for computing the inverse of
a matrix increases nearly at the rate $O(\vert \Delta \vert^3)$ 
for $\vert \Delta \vert > 100$, 
it increases too slowly for relatively small $\vert \Delta \vert$.
On the other hand, we can see from Table \ref{table:1} that the CPU time
of the proposed algorithm almost linearly increases in proportion to
$\vert \Delta \vert$ which follows the theoretical result in the
previous section.  
These are the reasons why the proposed algorithm is slower than
the direct computation for relatively small $\vert \Delta \vert$.  
When $\vert \Delta \vert \ge 200$, however, 
the computational cost of $((K^{t-1})_{DD})^{-1}$ is not ignorable and
the proposed algorithm shows a considerable reduction of computational
time.
In practice the performances for large models are more crucial.
In this sense the proposed algorithm is considered to be efficient.

\begin{table}[htbp]
 \centering
 \caption{CPU time per one iterative procedure for $|\Delta|$ cycle
 models}
 \label{table:1}
 \begin{tabular}{rcc} \hline
  $\vert \Delta \vert$ & Algorithm \ref{alg:1} & direct computation \\ 
  \hline 
  5 & 1.590 & 2.261\\
 10 & 3.442 & 2.523\\
 50 & 18.63 & 5.482\\
100 & 38.11 & 17.48\\
200 & 78.87 & 104.91\\
300 & 120.03 & 361.01\\ 
500 & 225.94& 1292.1\\
1000& 511.60& 6625.8\\ \hline
  & &  ($10^{-2}$ CPU time)\\
 \end{tabular}
\end{table}

\begin{table}[htbp]
 \centering
 \caption{CPU time for calculating 
 $\vert \Delta \vert \times \vert \Delta \vert$ inverse matrices}
 \label{table:2}
 \begin{tabular}{cc} \hline
  $\vert \Delta \vert$ & CPU time \\  \hline 
  5 & 0.027\\
  10 & 0.031\\
  50 & 0.058\\
  100 & 0.286\\
  200 & 1.789\\
  300 & 5.602\\
  500 & 24.666\\
  1000 & 207.66\\ \hline
  & ($10^{-2}$ CPU time)\\
\end{tabular}
\end{table}

\section{Concluding remarks}
In this article we discussed the localization to reduce the
computational burden of the IPS in two ways. 
We first showed that the decomposition into mp-subgraphs of the graph 
can localize the IPS.  
Next we proposed a localized algorithm of the iterative step
in the IPS by using the structure of a chordal extension of 
the graphical model for each mp-subgraph.  
The proposed algorithm costs $O(\vert \Delta \vert)$ in the case of 
cycle models and some numerical experiments confirmed the theory
for large models.

As mentioned in Section 1, the implementation of the IPS requires
enumeration of all maximal cliques of the graph and this enumeration 
has an exponential complexity.   
In addition, the proposed algorithm also requires some characteristics
of graphs, that is, a chordal extension, perfect sequences and perfect
elimination orders of the chordal extension.
In this sense, the application of the IPS may be limited. 
However in the case where the structure of the model is simple or
sparse, it may be feasible to obtain characteristics of graphs. 
In such cases, the proposed algorithm is considered to be effective.

\begin{center}
 {\bf Acknowledgment}
\end{center}
 The authors are grateful to two anonymous referees for constructive
 comments and suggestions which have led to improvements in the
 presentation of the paper.

\bibliographystyle{plain}
\bibliography{hara-takemura-IPS}
\end{document}